\documentclass[a4paper,12pt]{article}

\usepackage{amsmath,amssymb,amsthm}
\usepackage{abstract}
\usepackage{titling}
\usepackage{tabularx}
\usepackage[japanese,english]{babel}
\usepackage[table]{xcolor}
\usepackage{adjustbox}
\usepackage{subcaption}
\usepackage[top=31.75mm,bottom=31.75mm,left=31.75mm,right=31.75mm]{geometry}
\usepackage{graphicx}
\usepackage{mathtools}
\usepackage{pdflscape}
\usepackage{colortbl}
\usepackage{comment}
\usepackage{amsfonts}
\usepackage{bbm,bm}
\usepackage{physics}
\usepackage{float}
\usepackage[bottom]{footmisc}
\usepackage{cancel}
\usepackage{cases}
\usepackage{ascmac}
\usepackage{ulem} 
\usepackage{stmaryrd}
\usepackage{setspace}
\usepackage[labelfont=bf]{caption}
\usepackage{setspace}
\usepackage{hyperref}
\hypersetup{
    colorlinks=true,
    linkcolor=orange,
    filecolor=orange,      
    urlcolor=orange,
    citecolor=orange,
}

\usepackage{physics} 
\usepackage{titlesec}
\titleformat{\section}
  {\normalfont\large\bfseries}{\thesection}{1em}{}
  \titleformat{\subsection}
  {\normalfont\normalsize\bfseries}{\thesubsection}{1em}{}

\newenvironment{proofoftheorem}[1]{%
\par\noindent\textbf{Proof of Theorem #1.}\ }{}

\theoremstyle{definition}
\newtheorem{Assumption}{Assumption}
\newtheorem*{Discussion in section 5.4}{Discussion in section 5.4}
\newtheorem*{Remark*}{Remark}

\newtheorem{Remark}{Remark}
\newtheorem{Theorem}{Theorem}

\newtheorem*{Def}{Definition}

\newtheorem*{Fact}{Fact}

\usepackage[round]{natbib}
\usepackage{multibib}
\newcites{online}{References}

\title{Triple Instrumented Difference-in-Differences\thanks{I am grateful to my advisors, Daiji Kawaguchi and Ryo Okui, for their continued guidance and support. I gratefully acknowledge the support provided by the JSPS KAKENHI Grant JP 24KJ0817. All errors are my own.}}
\author{Sho Miyaji\thanks{Graduate School of Economics, The University of Tokyo, 7-3-1 Hongo, Bunkyoku, Tokyo 113-0033, Japan; Email: \href{mailto:shomiyaji-apple@g.ecc.u-tokyo.ac.jp}{shomiyaji-apple@g.ecc.u-tokyo.ac.jp}.}}

\date{\today}

\begin{document}
\onehalfspacing
\maketitle   
\begin{abstract}
In this paper, we formalize a triple instrumented difference-in-differences (DID-IV). In this design, a triple Wald-DID estimand, which divides the difference-in-difference-in-differences (DDD) estimand of the outcome by the DDD estimand of the treatment, captures the local average treatment effect on the treated. The identifying assumptions mainly comprise a monotonicity assumption, and the common acceleration assumptions in the treatment and the outcome. We extend the canonical triple DID-IV design to staggered instrument cases. We also describe the estimation and inference in this design in practice.
\end{abstract} 
\newpage
\section{Introduction}\label{sec1}
Instrumented difference-in-differences (DID-IV) is a method to estimate the effect of a treatment on an outcome, exploiting the timing variation of a policy shock as an instrument for treatment. In its canonical format, some units remain unexposed to the instrument during the two periods (unexposed group), while others become exposed in the second period (exposed group). The target parameter is the local average treatment effect on the treated, and the identifying assumptions mainly comprise a monotonicity assumption, and the parallel trends assumtpions in the treatment and the outcome between the two groups. In this design, the Wald-DID estimand, which scales the DID estimand of the outcome by the DID estimand of the treatment, captures the local average treatment effect on the treated (\cite{chasemartin2010-ch}, \cite{Hudson2017-tm}, \cite{Miyaji2023}).\par
DID-IV designs are widely used for causal inference across many fields in economics (e.g., \cite{Duflo2001-nh}, \cite{Black2005-aw}, \cite{Isen2017-py}), and are helpful when there is no control group or the parallel trends assumption in DID designs is not plausible in practice. For instance, to estimate the causal relationship between children and parents' education attainment, \cite{Black2005-aw} employ the DID-IV identification strategy, exploiting the timing variation of school reforms across municipalities as an instrument for parents' education attainment.\par
In empirical work, however, when researchers adopt DID-IV designs, they often leverage the group structure $A_{i} \in \{0,1\}$ in the data—such as a demographic characteristic—in addition to the timing variation of a policy shock (instrument). For instance, \cite{Deschenes2017-nh} estimate the effect of Nitrogen Oxides (NO$_{\text{x}}$) emissions on mortality rate using a DID-IV identification strategy, leveraging the NO$_{\text{x}}$ Budget Trading program as an instrument for these emissions. The program was implemented in participating states only during the summer months (May–September), but not in non-participating states and not in winter months (January–April or October–December). In other words, \cite{Deschenes2017-nh} construct an instrument based on three sources of variation: time (pre- or post-implementation of the program), state (participating or non-participating), and season (summer or winter).\par 
In this paper, we formalize the underlying identification strategy as a triple instrumented difference-in-differences. We define the target parameter and identifying assumptions in this design, and extend it staggered instrument cases. We also describe the estimation and inference in this design in practice.\par
First, we consider two periods settings, where a policy shock (instrument) is assigned only to one group at the second period (exposed group), while it is not assigned to other group over time (unexposed group), and it is only introduced to a particular group $A_i=1$ for $A_i \in \{0,1\}$ in an exposed group. In this setting, our target parameter is the local average treatment effect on the treated in group $A_i=1$; this parameter captures the treatment effects, for those who are the compliers in group $A_i=1$ and exposed group. The key identifying assumptions are (i) a monotonicity assumption and (ii) common acceleration assumptions in the treatment and the outcome. As in triple DID designs (\cite{Olden2022-uj}, \cite{Frohlich2019-cw}, \cite{Wooldridge2019-is}), the common acceleration assumption does not require parallel trends to hold within each group separately (i.e., among the exposed and unexposed groups when comparing group $A_i=1$ and $A_i=0$ ). Rather, it only requires that any bias due to a violation of parallel trends between $A_i=1$ and $A_i=0$ in an exposed group is offset by the corresponding bias in an unexposed group. We show that in this design, the triple Wald-DID estimand, which scales the difference-in-difference-in-differences (DDD) estimand of the outcome by the DDD estimand of the treatment, captures the local average treatment effect on the treated in group $A_i=1$.\par
We extend the canonical triple DID-IV design to multiple periods settings, where the instrument (policy shock) is adopted at different points in time across groups (e.g., states or counties that each unit belongs to), and it is only applied to some demographic group $A_i=1$ for $A_i \in \{0,1\}$ in each group. We call this the staggered triple DID-IV design, and define the target parameter and identifying assumptions. First, we partition groups, such as states or counties, into mutually exclusive and exhaustive cohorts, based on the initial adoption date of the instrument (policy shock). We then define our target parameter as the cohort specific local average treatment effect on the treated (CLATT) in group $A_i=1$; this parameter measures the treatment effects, for those who are the compliers at a given relative period $l$ in group $A_i=1$ and cohort $c$. Finally, we extend the identifying assumptions in the canonical triple DID-IV design to multiple periods settings, and show that each triple Wald-DID estimand captures the CLATT in group $A_i=1$ and cohort $c$ under this design.\par
We describe the estimation and inference in triple DID-IV design in practice. In two periods settings, one can estimate the sample analog of the triple Wald-DID estimand by running an IV regression that implements the triple DID regression in both the first stage and the reduced form. In multiple periods settings, the estimation proceeds in two steps. First, we restrict the data to include only two periods (before and after the policy shock) and two cohorts, with one cohort serving as a control group. Second, we run the IV regression on each such data subset, again applying the triple DID regression in both the first stage and the reduced form. In both the two-period and multiple-period cases, we study the asymptotic properties of the sample analog of the triple Wald-DID estimand, showing that it is consistent and asymptotically normal for the target parameter.\par
This paper is related to the recent DID-IV literature (\cite{chasemartin2010-ch}, \cite{Hudson2017-tm}, \cite{Miyaji2023}), and contributes to this literature  by introducing the common acceleration assumption, considered in triple DID designs (\cite{Olden2022-uj}, \cite{Frohlich2019-cw}, \cite{Wooldridge2019-is}), into the canonical DID-IV framework. In the DID-IV literature, \cite{chasemartin2010-ch} is the first to formalize the DID-IV design, showing that a Wald-DID estimand identifies the LATET under a monotonicity assumption, and the parallel trends assumptions in the treatment and the outcome. \cite{Hudson2017-tm} extends it to the case of a non-binary, ordered treatment. In this paper, given the additional group structure $A_i \in \{0,1\}$, we consider the triple Wald-DID estimand, and adopt the common acceleration assumptions in the treatment and the outcome, following the triple DID literature (\cite{Olden2022-uj}, \cite{Frohlich2019-cw}, \cite{Wooldridge2019-is}). Our triple DID-IV designs are robust to the potential violation of the parallel trends assumptions in the treatment and the outcome, and thus serve as an alternative to the canonical DID-IV design in practice.\par
The rest of the paper is organized as follows. Section \ref{sec2} establishes triple DID-IV designs in two periods settings. Section \ref{sec3} extends the canonical triple DID-IV design to multiple periods settings with staggered instrument. Section \ref{sec4} describes the estimation and inference in triple DID-IV designs. Section \ref{sec5} concludes. All proofs are given in the Appendix.
\section{Triple DID-IV design}\label{sec2}
\subsection{Set up}\label{sec2.1}
We consider panel data settings with two periods and $N$ units. Let $Y_{i,t}$ be the outcome and $D_{i,t} \in \{0,1\}$ be the treatment for unit $i$ and time $t$. Let $Z_{i,t} \in \{0,1\}$ be the instrument for unit $i$ and time $t$. Let $D_{i}=(D_{i,0}, D_{i,1})$ be the treatment path and $Z_{i}=(Z_{i,0}, Z_{i,1})$ be the instrument path for unit $i$. In the following, let $\mathcal{S}(R)$ be the support for any random variable $R$.\par
To motivate our setting considered in this paper, suppose that a researcher is interested in estimating the effect of a treatment on an outcome, but is concerned about endogeneity. To address this issue, the researcher constructs an instrument based on a policy shock that is assigned only in the second period. This shock is introduced exclusively to one group (exposed group) and not to the other (unexposed group). Furthermore, within the exposed group, the policy shock applies only to individuals for whom $A_i=1$ ($A_i \in \{0,1\}$)—for instance, only to women.\par 
The following assumption describes the above assignment process of the instrument. 
\begin{Assumption}[Triple DID-IV design]\label{sec2as0}
$Z_{i,0}=0$\hspace{3mm}\text{for all $i$}.
\begin{align*}
Z_{i,1}=
\begin{cases}
1 & \text{if}\hspace{2mm} A_i=1\hspace{2mm} \text{and}\hspace{2mm} C_i=1\hspace{2mm}\\
0 & \text{if otherwise}
\end{cases}
\end{align*}
\end{Assumption}\par
Here, $C_i \in \{0,1\}$ is the group indicator that takes one if unit $i$ belongs to the exposed group. Note that we do not impose any restrictions on the assignment process of the treatment. Therefore, the treatment path can take four values, that is, we have $D_i \in \{(0,0),(0,1),(1,0),(1,1)\}$.\par\par 
To make the situation more concrete, we provide an empirical example. \cite{Deschenes2017-nh} estimate the effect of Nitrogen Oxides (NO$_{\text{x}}$) emissions on mortality rate, using the NO$_{\text{x}}$ Budget Trading program as an instrument for these emissions. The program was implemented in participating states only during the summer months (May–September), but not in non-participating states and not in winter months (January–April or October–December). In their setup, the group variable $C_i$ is an indicator for whether the state (unit $i$ belongs to) participated in the program (participating vs. non-participating), and the group variable $A_i$ is an indicator for the season (summer vs. winter).\par
We now introduce the potential outcomes framework. Let $Y_{i,t}(d,z)$ be the potential outcome for unit $i$ and time $t$ if $D_i=d$ and $Z_i=z$. Let $D_{i,t}(z)$ be the potential treatment choice for unit $i$ and time $t$ if $Z_i=z$. Since the treatment and the instrument takes only two values in each period, we can write the observed treatment and the outcome as follows:
\begin{align*}
 &D_{i,t}=\sum_{z \in \mathcal{S}(Z)}\mathbf{1}\{Z_i=z\}D_{i,t}(z),\hspace{3mm}Y_{i,t}=\sum_{z \in \mathcal{S}(Z)}\sum_{d \in \mathcal{S}(D)}\mathbf{1}\{Z_i=z, D_i=d\}Y_{i,t}(d,z).
 \end{align*}\par
 Here, we assume the no carryover assumption on potential outcomes.
 \begin{Assumption}[No carryover assumption]\label{sec2as1}
\begin{align*}
\forall z \in \mathcal{S}(Z),\forall d\in \mathcal{S}(D),Y_{i,0}(d,z)=Y_{i,0}(d_0,z),Y_{i,1}(d,z)=Y_{i,1}(d_1,z),
\end{align*}
where $d=(d_0,d_1)$ is the generic element of the treatment path $D_i$.
\end{Assumption}
This assumption requires that the potential outcome $Y_{i,t}(d,z)$ depends only on the current treatment status $d_t$ for all $z \in \mathcal{S}(Z)$ and all $t \in \{0,1\}$. In the DID literature, \cite{De_Chaisemartin2020-dw} and \cite{Imai2021-dn} impose the similar assumption under non-staggered treatment settings.\par
 Finally, we introduce the group variable $G_i^{Z}=(D_{i,1}(0,0),D_{i,1}(0,1))$. This group variable describes the response of the treatment to the instrument path $Z_i$ in time $t=1$. Following the terminology in \cite{Imbens1994-qy}, we call $G_i^{Z}=(0,0) \equiv NT^{Z}$ as the never-takers, $G_i^{Z}=(0,1) \equiv CM^{Z}$ as the compliers, $G_i^{Z}=(1,0) \equiv DF^{Z}$ as the defiers, and $G_i^{Z}=(1,1) \equiv AT^{Z}$ as the always-takers.\par
 Based on the notation developed in this section, the next section defines the target parameter in triple DID-IV design.
\subsection{The target parameter in triple DID-IV design}
In triple DID-IV design, our target parameter is the local average treatment effect on the treated (LATET) at time $t=1$ and group $A_i=1$ defined below.
\begin{Def}
The local average treatment effect on the treated (LATET) at time $t=1$ and group $A_i=1$ is
\begin{align*}
LATET &\equiv E[Y_{i,1}(1)-Y_{i,1}(0)|C_i=1,A_i=1,D_{i,1}((0,1)) > D_{i,1}((0,0))]\\
&=E[Y_{i,1}(1)-Y_{i,1}(0)|C_i=1,A_i=1,CM^{Z}].
\end{align*}\par
\end{Def}
This parameter measures the treatment effects, for those who are the compliers ($CM^{Z}$) at time $t=1$ in group $C_i=1$, and $A_i=1$. The LATET is also defined in the recent DID-IV literature (\cite{chasemartin2010-ch}, \cite{Miyaji2023}). The difference here is that it is conditional on group variable $A_i=1$ in triple DID-IV design.\par
\begin{Remark}
When we have a non-binary, ordered treatment $D_{i,t} \in \{0,\dots,J\}$, our target parameter is the average causal response on the treated (ACRT) given $A_i=1$ defined below.
\begin{Def}
The average causal response on the treated (ACRT) is
\begin{align*}
ACRT \equiv \sum_{j=1}^{J}w_j \cdot E[Y_{i,1}(j)-Y_{i,1}(j-1)|D_{i,1}((0,1)) \geq j > D_{i,1}((0,0)), C_i=1, A_i=1],
\end{align*}
where the weight $w_j$ is:
\begin{align*}
w_j=\frac{Pr(D_{i,1}((0,1)) \geq j > D_{i,1}((0,0))|C_i=1, A_i=1)}{\sum_{j=1}^{J} Pr(D_{i,1}((0,1)) \geq j > D_{i,1}((0,0))|C_i=1, A_i=1)}.
\end{align*}
This parameter is the conditional version of the average causal response considered in \cite{Angrist1995-ij}. In the recent DID-IV literature, \cite{Miyaji2023} considers the similar parameter.
\end{Def}
\end{Remark}
\subsection{The identifying assumptions in triple DID-IV design}\label{sec2.3}
In this section, we establish the identifying assumptions in triple DID-IV design. In this design, we exploit the group structure $A_i \in \{0,1\}$ in the data, in addition to the timing variation of a policy shock. We therefore consider the following estimand, calling it the triple Wald-DID estimand:
\begin{align*}
w_{DID} &= \frac{DID_{Y,C=1,A=1}-DID_{Y,C=0,A=1}
-(DID_{Y,C=1,A=0}-DID_{Y,C=0,A=0})}{DID_{D,C=1,A=1}-DID_{D,C=0,A=1}
-(DID_{D,C=1,A=0}-DID_{D,C=0,A=0})},
\end{align*}
where
\begin{align*}
&DID_{Y,C=c,A=a}=E[Y_{i,1}-Y_{i,0}|C_i=c,A_i=a],\\
&DID_{D,C=c,A=a}=E[D_{i,1}-D_{i,0}|C_i=c,A_i=a],
\end{align*}
for $a \in \{0,1\}$ and $c \in \{0,1\}$.\par
Note that this estimand scales the difference-in-difference-in-differences (DDD) estimand of the outcome by the DDD estimand of the treatment, a natural extension of the Wald-DID estimand considered in the recent DID-IV literature (\cite{chasemartin2010-ch}, \cite{Hudson2017-tm}, \cite{Miyaji2023}, and \cite{Miyaji2023-tw}).\par
We consider the following identifying assumptions for the triple Wald-DID estimand to identify the LATET at time $t=1$ and group $A_i=1$.\par
\begin{Assumption}[Exclusion restriction]\label{sec2as2}
\begin{align*}
\forall z \in \mathcal{S}(Z),\forall d\in \mathcal{S}(D),\forall t\in \{0,1\},Y_{i,t}(d,z)=Y_{i,t}(d).
\end{align*}
\end{Assumption}
This assumption requires that the potential outcome $Y_{i,t}(d,z)$ depends only on the treatment path $d \in \mathcal{S}(D)$ and does not depend on the instrument path $z \in \mathcal{S}(Z)$ for all $t \in \{0,1\}$.\par 
Given Assumptions \ref{sec2as1}-\ref{sec2as2}, we can write the observed outcome $Y_{i,t}$ as follows:
\begin{align*}
Y_{i,t}=D_{i,t}Y_{i,t}(1)+(1-D_{i,t})Y_{i,t}(0).
\end{align*}
For any $z \in \mathcal{S}(Z)$, let $Y_{i,t}(D_{i,t}(z))$ be the outcome if $Z_i=z$:
\begin{align*}
Y_{i,t}(D_{i,t}(z))=D_{i,t}(z)Y_{i,t}(1)+(1-D_{i,t}(z))Y_{i,t}(0).
\end{align*}
Following the terminology in \cite{Miyaji2023}, we call $Y_{i,t}(D_{i,t}((0,0)))$ as unexposed outcome, and $Y_{i,t}(D_{i,t}((0,1)))$ as exposed outcome for unit $i$ and time $t$.\par
\begin{Assumption}[Monotonicity assumption in time $t=1$]\label{sec2as3}
\begin{align*}
D_{i,1}((0,1)) \geq D_{i,1}((0,0))\hspace{3mm}\text{with probability 1}.
\end{align*}
\end{Assumption}
This assumption requires that the instrument path affects the potential treatment choice at period $t=1$ in one direction, excluding the defiers ($DF^{Z}$). This assumption is common in the IV literature (e.g., \cite{Imbens1994-qy}).
\begin{Assumption}[No anticipation in the first stage]
\label{sec2as4}
\begin{align*}
D_{i,0}((0,1))=D_{i,0}((0,0))\hspace{2mm}\text{for all $i$ with $C_i=1$ and $A_i=1$}.
\end{align*}
\end{Assumption}
This assumption requires that the potential treatment choice before the exposure to the instrument is equal to the one without instrument in group $A_i=1$ in an exposed group. The no anticipation assumption is commonly imposed on the potential outcome in the recent DID literature (e.g., \cite{Callaway2021-wl}, \cite{Athey2022-uo}, \cite{Roth2023-ig}). The difference here is that it is imposed on the potential treatment choice in DID-IV designs (\cite{Miyaji2023}).\par
\begin{Assumption}[Relevance condition]
\label{sec2as5}
\begin{align*}
DID_{D,C=1,A=1}-DID_{D,C=0,A=1}
-(DID_{D,C=1,A=0}-DID_{D,C=0,A=0}) > 0.
\end{align*}
\end{Assumption}
This assumption is the relevance condition in the first stage, ensuring that the triple Wald-DID estimand is well defined.\par
Finally, we make the common acceleration assumptions in the treatment and the outcome. These assumptions are unique in triple DID-IV design.\par
\begin{Assumption}[Common acceleration assumption in the treatment]
\label{sec2as6}
\begin{align*}
&E[D_{i,1}((0,0))-D_{i,0}((0,0))|C_i=1,A_i=1]\\
&-E[D_{i,1}((0,0))-D_{i,0}((0,0))|C_i=1,A_i=0]\\
&\hspace{4cm} = \hspace{4cm} \\
&E[D_{i,1}((0,0))-D_{i,0}((0,0))|C_i=0,A_i=1]\\
&-E[D_{i,1}((0,0))-D_{i,0}((0,0))|C_i=0,A_i=0].
\end{align*}
\end{Assumption}

\begin{Assumption}[Common acceleration assumption in the outcome]
\label{sec2as7}
\begin{align*}
&E[Y_{i,1}(D_{i,1}(0,0))-Y_{i,0}(D_{i,0}((0,0)))|C_i=1,A_i=1]\\
&-E[Y_{i,1}(D_{i,1}(0,0))-Y_{i,0}(D_{i,0}((0,0)))|C_i=1,A_i=0]\\
&\hspace{4cm} = \hspace{4cm} \\
&E[Y_{i,1}(D_{i,1}(0,0))-Y_{i,0}(D_{i,0}((0,0)))|C_i=0,A_i=1]\\
&-E[Y_{i,1}(D_{i,1}(0,0))-Y_{i,0}(D_{i,0}((0,0)))|C_i=0,A_i=0].
\end{align*}
\end{Assumption}
Assumption \ref{sec2as6} and \ref{sec2as7} require that the bias arising from the parallel trends assumption in the treatment and the outcome between group $0$ and $1$ is the same between exposed and unexposed groups. In triple DID designs, this assumption is made on the untreated outcome (\cite{Frohlich2019-cw}, \cite{Wooldridge2019-is}, \cite{Olden2022-uj}).\par
The following theorem shows that the triple Wald-DID estimand identifies the LATET in time $t=1$ and group $A_i=1$ under Assumptions \ref{sec2as0}-\ref{sec2as7}.
\begin{Theorem}\label{sec2.theorem1}
If Assumptions \ref{sec2as0}-\ref{sec2as7} hold, the triple Wald-DID estimand $w_{DID}$ captures the LATET at time $t=1$ and group $A_i=1$; that is,
\begin{align*}
w_{DID}=E[Y_{i,1}(1)-Y_{i,1}(0)|C_i=1,A_i=1,CM^{Z}].
\end{align*}
\end{Theorem}
\begin{proof}
See Appendix.
\end{proof}
\begin{Remark}
When the treatment $D_{i,t}$ is non-binary, we have the following theorem.
\begin{Theorem}
Suppose that Assumptions \ref{sec2as0}-\ref{sec2as7} hold, which replace the binary treatment with non-binary one. Then, the triple Wald-DID estimand $w_{DID}$ captures the ACRT at time $t=1$ and group $A_i=1$; that is,
\begin{align*}
w_{DID}=ACRT.
\end{align*}
\end{Theorem}
\begin{proof}
This theorem holds by combining the proof in Theorem \ref{sec2.theorem1} with the proof in Theorem 4 in \cite{Miyaji2023}. Thus, we omit it for brevity.
\end{proof}
\end{Remark}
\section{Triple DID-IV design with multiple time periods}\label{sec3}
In this section, we extend the canonical triple DID-IV design to multiple period settings with staggered instrument, calling it a staggered triple DID-IV design. 
\subsection{Set up}\label{sec3.1}
We consider panel data settings with $N$ units, observed in each period $t \in \{1,\dots,T\}$. Let $D_i=(D_{i,1},\dots,D_{i,T})$ be the treatment path and $Z_i=(Z_{i,1},\dots,Z_{i,T})$ be the instrument path for unit $i$.\par
For illustrative purposes, suppose that each unit belongs to a specific state, and that each state adopts a new policy (serving as the instrument) at different points in time. Once a state has adopted the policy, it remains in effect thereafter. Moreover, assume that within each state, the policy is introduced only to a particular demographic group $A_i=1$, such as females or males.\par
To describe the above situation, we first make the following assumption on the assignment process of the instrument.
\begin{Assumption}[Staggered instrument adoption]
\label{sec3as1}
$Z_{i,1}=0$\hspace{2mm}\text{for all $i$}.\\
\text{For each}\hspace{1mm} $t \in \{2,\dots,T\}$, $Z_{i,t-1} \leq Z_{i,t}$ \hspace{2mm}\text{for all $i$}.
\end{Assumption}
This assumption requires that the instrument in time $t=1$ is equal to zero for all $i$, excluding the already exposed units. Further, it requires that once units start receiving the instrument, they remain exposed to that instrument, which we call the staggered instrument adoption.\footnote{In the recent DID-IV literature, \cite{Miyaji2023} considers the same assumption. For the case of non-staggered instrument, see \cite{dechaisemartin2024differenceindifferencesestimatorstreatmentscontinuously}.}
\par 
Here, we introduce the cohort variable $C_i \in \{2,\dots,T,\infty\}$: $C_i=c$ if unit $i$ belongs to the states that receive the policy shock (instrument) at time $t=c$. We set $C_i=\infty$ if unit $i$ belongs to the states that are never exposed to the instrument.\par
Next, we assume that the instrument $Z_{i,t}$ takes one only if unit $i$ belongs to group $A_i=1$ in each cohort $C_i=c$.\par
\begin{Assumption}[Triple staggered DID-IV design]
\label{sec3as_triple}
\text{For each}\hspace{1mm} $i$\hspace{1mm}  and\hspace{1mm} $t \in \{1,\dots,T\}$, \begin{align*}
Z_{i,t}=
\begin{cases}
1 & \text{if}\hspace{2mm} A_i=1\hspace{2mm} \text{and}\hspace{2mm} C_i=c\hspace{1mm} (t \geq c)\hspace{2mm}\\
0 & \text{if otherwise}
\end{cases}
\end{align*}
\end{Assumption}
Let $E_{i}=\min\{t: Z_{i,t}=1\}$ be the initial adoption date of the instrument for unit $i$. We set $E_i=\infty$ if unit $i$ is never exposed to the instrument. Then, under Assumptions \ref{sec3as1}-\ref{sec3as_triple}, we have
\begin{align*}
E_{i}=
\begin{cases}
c & \text{if}\hspace{2mm} A_i=1\hspace{2mm} \text{and}\hspace{2mm} C_i=c\\
\infty & \text{if otherwise}
\end{cases}
\end{align*}\par
We rewrite the potential treatment choice $D_{i,t}(z)$, using the initial adoption date of the instrument $E_{i}$. Let $D_{i,t}^{c}$ be the potential treatment choice for unit $i$ in time $t$ if $E_i=c$. Let $D_{i,t}^{\infty}$ be the potential treatment choice for unit $i$ in time $t$ if $E_i=\infty$. In the following, we refer to $D_{i,t}^{\infty}$ as “never exposed treatment”. Then, we can express the observed treatment choice $D_{i,t}$ as follows:
\begin{align*}
D_{i,t}=D_{i,t}^{\infty}+\sum_{2 \leq c \leq T}(D_{i,t}^{c}-D_{i,t}^{\infty})\cdot \mathbf{1}\{E_i=c\}.
\end{align*}
Here, as in \cite{Miyaji2023}, we define $D_{i,t}-D_{i,t}^{\infty}$ to be the effect of the instrument on the potential treatment choice for unit $i$ in time $t$, calling it the “individual exposed effect in the first stage”.\footnote{\cite{Sun2021-rp} and \cite{Callaway2021-wl} define the effect of a treatment on an outcome in a similar fashion under staggered DID settings.}\par
Note that in contrast to staggered instrument adoption, we allow the general adoption process of the treatment, i.e., we allow that the treatment can turn on/off over time.\par 
Similar to section \ref{sec2}, we impose the no carry over assumption on potential outcomes.
\begin{Assumption}[No carryover assumption in multiple time periods]
\label{sec3as2}
\begin{align*}
\forall z \in \mathcal{S}(Z), \forall d\in \mathcal{S}(D), \forall t\in \{1,\dots,T\},Y_{i,t}(d,z)=Y_{i,t}(d_t,z)\hspace{3mm}\text{for all $i$},
\end{align*}
where $d=(d_0,\dots,d_t,\dots,d_T)$ is the generic element of the treatment path $D_i$.
\end{Assumption}

\par
Finally, we introduce the group variable $G_{i,e,t} \equiv (D_{i,t}^{\infty},D_{i,t}^{c})$ ($t \geq c$). This group variable expresses the response of the potential treatment choice at time $t$ to the instrument path $z$. Following section \ref{sec2.1}, we call $G_{i,c,t}=(0,0) \equiv NT_{c,t}$ as the never-takers, $G_{i,c,t}=(0,1) \equiv CM_{c,t}$ as the compliers, $G_{i,c,t}=(1,0) \equiv DF_{c,t}$ as the defiers and $G_{i,c,t}=(1,1) \equiv AT_{c,t}$ as the always-takers in period $t$ and the initial exposure date $c$. 
\subsection{The target parameter in staggered triple DID-IV design}\label{sec3.2}
In staggered triple DID-IV design, our target parameter is the cohort specific local average treatment effect on the treated (CLATT) given group $A_i=1$ defined below.
\begin{Def}
The cohort specific local average treatment effect on the treated (CLATT) in a relative period $l$ from the initial adoption of the instrument is 
\begin{align*}
CLATT_{c,c+l}&=E[Y_{i,c+l}(1)-Y_{i,c+l}(0)|C_i=c, A_i=1, D_{i,c+l}^{c} > D_{i,c+l}^{\infty}]\\
&=E[Y_{i,c+l}(1)-Y_{i,c+l}(0)|C_i=c, A_i=1,CM_{c,c+l}].
\end{align*}
\end{Def}
This parameter captures the treatment effects, for those who are the compliers in time $c+l$ in group $A_i=1$ and cohort $C_i=c$. This parameter is also considered in \cite{Miyaji2023}, but it is conditional on $A_i=1$ in triple DID-IV settings.

\begin{Remark} When we have a non-binary, ordered treatment in staggered triple DID-IV design, our target parameter is the cohort specific average causal response on the treated (CACRT) given group $A_i=1$ defined below.
\begin{Def} The cohort specific average causal response on the treated (CACRT) at a given relative period $l$ from the initial adoption of the instrument is
\begin{align*}
CACRT_{c,c+l} \equiv \sum_{j=1}^{J}w^{c}_{c+l,j} \cdot E[Y_{i,c+l}(j)-Y_{i,c+l}(j-1)|C_i=c, A_i=1, D_{i,c+l}^{c} \geq j > D_{i,c+l}^{\infty}],
\end{align*}
where the weight $w^{c}_{c+l,j}$ is:
\begin{align*}
w^{c}_{c+l,j}=\frac{Pr(D_{i,c+l}^{c} \geq j > D_{i,c+l}^{\infty}|C_i=c, A_i=1)}{\sum_{j=1}^{J} Pr(D_{i,c+l}^{c} \geq j > D_{i,c+l}^{\infty}|C_i=c, A_i=1)}.
\end{align*}
\end{Def}
Note that the CACRT is also defined in \cite{Miyaji2023}. The difference here is that it is conditional on $A_i=1$ in triple DID-IV design.
\end{Remark}

\subsection{The identifying assumptions in staggered triple DID-IV design}\label{sec3.3}
In this section, we formalize the identifying assumptions in staggered triple DID-IV design.\par
In staggered triple DID-IV design, we consider the following estimand to identify each $CLATT_{c,c+l}$:
\begin{align*}
w^{DID}_{c,l}=\frac{DID^{l}_{Y,C=c,A=1}-DID^{l}_{Y,C=\infty,A=1}-(DID^{l}_{Y,C=c,A=0}-DID^{l}_{Y,C=\infty,A=0})}{DID^{l}_{D,C=c,A=1}-DID^{l}_{D,C=\infty,A=1}-(DID^{l}_{D,C=c,A=0}-DID^{l}_{D,C=\infty,A=0})},
\end{align*}
where
\begin{align*}
&DID^{l}_{Y,C=c,A=a}=E[Y_{i,c+l}-Y_{i,c-1}|C_i=c,A_i=a],\\
&DID^{l}_{D,C=c,A=a}=E[D_{i,c+l}-D_{i,c-1}|C_i=c,A_i=a],
\end{align*}
for $a \in \{0,1\}$, $c \in \{2,\dots,T,\infty\}$ and $l \in \{0,\dots,T-c\}$. Note that this estimand is the triple Wald-DID esitmand, where the pre-exposed period is $c-1$ and the control group is $C_i=\infty$, the never exposed cohort.\par
The following assumptions are sufficient for each $w^{DID}_{c,l}$ to capture the $CLATT_{c,c+l}$. 

\begin{Assumption}[Exclusion restriction in multiple time periods]
\label{sec3as3}
\begin{align*}
\forall z \in \mathcal{S}(Z),\forall d \in \mathcal{S}(D),\forall t \in \{1,\dots,T\}, Y_{i,t}(d,z)=Y_{i,t}(d)\hspace{3mm}\text{for all $i$}.
\end{align*}
\end{Assumption}
Assumption \ref{sec3as3} is the exclusion restriction in multiple period settings. Assumption \ref{sec3as2} and \ref{sec3as3} imply that we can write $Y_{i,t}=D_{i,t}Y_{i,t}(1)+(1-D_{i,t})Y_{i,t}(0)$. Following section \ref{sec2.3}, we introduce the outcome in time $t$ if unit $i$ is exposed to the instrument path $z$:
\begin{align*}
Y_{i,t}(D_{i,t}(z))=D_{i,t}(z)Y_{i,t}(1)+(1-D_{i,t}(z))Y_{i,t}(0).
\end{align*}
Since $D_{i,t}(z)$ can be characterized by the initial adoption date of the instrument $E_i$, we can also rewrite $Y_{i,t}(D_{i,t}(z))$: let $Y_{i,t}(D_{i,t}^{c})$ be the outcome in time $t$ if unit $i$ is first exposed to the instrument in time $t=c$. Let $Y_{i,t}(D_{i,t}^{\infty})$ be the outcome in time $t$ if unit $i$ is never exposed to the instrument. 
\begin{Assumption}[Monotonicity assumption in multiple time periods]
\label{sec3as4}
\begin{align*}
\forall c\in \{2,\dots,T\},\hspace{2mm}\forall t \geq c,\hspace{2mm}, D_{i,t}^{c} \geq D_{i,t}^{\infty}=1\hspace{2mm}\text{for all $i$}.
\end{align*}
\end{Assumption}
This assumption requires that the individual exposed effect in the first stage, $D_{i,t}-D_{i,t}^{\infty}$, is non-negative after the exposure to the instrument. This assumption rules out the existence of the defiers $DF_{c,t}$ for all $c \in \{2,\dots,T\}$ and $t \geq c$.

\begin{Assumption}[No anticipation in the first stage]
\label{sec3as5}
\begin{align*}
\forall c\in \{2,\dots,T\},\hspace{2mm}\forall t < c,\hspace{2mm}D_{i,t}^{c}=D_{i,t}^{\infty}\hspace{2mm}\hspace{3mm}\text{for all $i$}.
\end{align*}
\end{Assumption}
This assumption requires that the instrument does not affect the potential treatment choice before the exposure to that instrument. In the DID-IV literature, \cite{Miyaji2023} imposes the same assumption.\footnote{In the DID literature, \cite{Callaway2021-wl} and \cite{Sun2021-rp} make the similar assumption on the potential outcome under staggered DID settings.}
\begin{Assumption}[Relevance condition based on a never  exposed cohort]
\label{sec3as6}
\begin{align*}
&\text{For each}\hspace{2mm}c\in \{2,\dots,T\}\hspace{2mm}\text{and}\hspace{2mm}l \in \{0,\dots,T-c\},\\ 
&DID^{l}_{D,C=c,A=1}-DID^{l}_{D,C=\infty,A=1}-(DID^{l}_{D,C=c,A=0}-DID^{l}_{D,C=\infty,A=0}) > 0.
\end{align*}
\end{Assumption}
This assumption guarantees that each triple Wald-DID estimand $w^{DID}_{c,l}$ is well defined.\par
Finally, we make the common acceleration assumptions in the treatment and the outcome based on a never exposed cohort. 
\begin{Assumption}[Common acceleration assumption in the treatment based on a never exposed cohort]
\label{sec3as7}
\begin{align*}
&\text{For each}\hspace{2mm}c\in \{2,\dots,T\}\hspace{2mm}\text{and}\hspace{2mm}t \in \{2,\dots,T\}\hspace{2mm}\text{such that}\hspace{2mm}t \geq c,\\ 
&E[D_{i,t}^{\infty}-D_{i,t-1}^{\infty}|C_i=c,A_i=1]-E[D_{i,t}^{\infty}-D_{i,t-1}^{\infty}|C_i=c,A_i=0]\\
&\hspace{5.5cm} = \hspace{5.5cm} \\
&E[D_{i,t}^{\infty}-D_{i,t-1}^{\infty}|C_i=\infty,A_i=1]-E[D_{i,t}^{\infty}-D_{i,t-1}^{\infty}|C_i=\infty,A_i=0].
\end{align*}
\end{Assumption}

\begin{Assumption}[Common acceleration assumption in the outcome based on a never exposed cohort]
\label{sec3as8}
\begin{align*}
&\text{For each}\hspace{2mm}c\in \{2,\dots,T\}\hspace{2mm}\text{and}\hspace{2mm}t \in \{2,\dots,T\}\hspace{2mm}\text{such that}\hspace{2mm}t \geq c,\\ 
&E[Y_{i,t}(D_{i,t}^{\infty})-Y_{i,t-1}(D_{i,t-1}^{\infty})|C_i=c,A_i=1]-E[Y_{i,t}(D_{i,t}^{\infty})-Y_{i,t-1}(D_{i,t-1}^{\infty})|C_i=c,A_i=0]\\
&\hspace{8cm} = \hspace{8cm} \\
&E[Y_{i,t}(D_{i,t}^{\infty})-Y_{i,t-1}(D_{i,t-1}^{\infty})|C_i=\infty,A_i=1]-E[Y_{i,t}(D_{i,t}^{\infty})-Y_{i,t-1}(D_{i,t-1}^{\infty})|C_i=\infty,A_i=0].
\end{align*}
\end{Assumption}

The following theorem shows that under Assumptions \ref{sec3as1}-\ref{sec3as8}, each triple Wald-DID estimand $w^{DID}_{c,l}$ captures the $CLATT_{c,c+l}$.
\begin{Theorem}\label{sec3.theorem2}
If Assumptions \ref{sec3as1}-\ref{sec3as8} hold, each triple Wald-DID estimand $w^{DID}_{c,l}$ captures the $CLATT_{c,c+l}$; that is,
\begin{align*}
w^{DID}_{c,l}=E[Y_{i,c+l}(1)-Y_{i,c+l}(0)|C_i=c, A_i=1,CM_{c,c+l}],
\end{align*}
for each $c\in \{2,\dots,T\}$ and $l \in \{0,\dots,T-c\}$.
\end{Theorem}
\begin{proof}
See Appendix.
\end{proof}
Note that when there exists no never exposed cohort $C_i=\infty$, one can intead consider the following estimand:
\begin{align*}
w^{DID}_{c,l,m}=\frac{DID^{l}_{Y,C=c,A=1}-DID^{l}_{Y,m,A=1}-(DID^{l}_{Y,C=c,A=0}-DID^{l}_{Y,m,A=0})}{DID^{l}_{D,C=c,A=1}-DID^{l}_{D,m,A=1}-(DID^{l}_{D,C=c,A=0}-DID^{l}_{D,m,A=0})},
\end{align*}
where
\begin{align*}
&DID^{l}_{Y,m,A=a}=E[Y_{i,c+l}-Y_{i,c-1}|C_i=\max\{C_i\},A_i=a],\\
&DID^{l}_{D,m,A=a}=E[D_{i,c+l}-D_{i,c-1}|C_i=\max\{C_i\},A_i=a],
\end{align*}
for $a \in \{0,1\}$, $c \in \{2,\dots,\max\{C_i\}-1\}$ and $l \in \{0,\max\{C_i\}-1-c\}$. In this estimand, the control cohort is $C_i=\max\{C_i\}$, the last exposed cohort.\par
If we consider the above estimand $w^{DID}_{c,l,m}$, we can replace Assumptions \ref{sec3as6}-\ref{sec3as8} with Assumptions \ref{sec3as9}-\ref{sec3as11} below.

\begin{Assumption}[Relevance condition based on a last exposed cohort]
\label{sec3as9}
\begin{align*}
&\text{For each}\hspace{2mm}c\in \{2,\dots,\max\{C_i\}-1\}\hspace{2mm}\text{and}\hspace{2mm}l \in \{0,\max\{C_i\}-1-c\},\\ 
&DID^{l}_{D,C=c,A=1}-DID^{l}_{D,m,A=1}-(DID^{l}_{D,C=c,A=0}-DID^{l}_{D,m,A=0}) > 0.
\end{align*}
\end{Assumption}

\begin{Assumption}[Common acceleration assumption in the treatment based on a last exposed cohort]
\label{sec3as10}
\begin{align*}
&\text{For each}\hspace{2mm}c\in \{2,\dots,\max\{C_i\}-1\}\hspace{2mm}\text{and}\hspace{2mm}t\hspace{2mm}\text{such that}\hspace{2mm}c \leq t \leq \max\{C_i\}-1,\\ 
&E[D_{i,t}^{\infty}-D_{i,t-1}^{\infty}|C_i=c,A_i=1]-E[D_{i,t}^{\infty}-D_{i,t-1}^{\infty}|C_i=c,A_i=0]\\
&\hspace{5.5cm} = \hspace{5.5cm} \\
&E[D_{i,t}^{\infty}-D_{i,t-1}^{\infty}|C_i=\max\{C_i\},A_i=1]-E[D_{i,t}^{\infty}-D_{i,t-1}^{\infty}|C_i=\max\{C_i\},A_i=0].
\end{align*}
\end{Assumption}

\begin{Assumption}[Common acceleration assumption in the outcome based on a last exposed cohort]
\label{sec3as11}
\begin{align*}
&\text{For each } c \in \{2,\dots,\max\{C_i\}-1\} \text{ and } t \text{ such that } c \leq t \leq \max\{C_i\}-1, \\
&E[Y_{i,t}(D_{i,t}^{\infty}) - Y_{i,t-1}(D_{i,t-1}^{\infty}) \mid C_i = c, A_i = 1] \\
-&E[Y_{i,t}(D_{i,t}^{\infty}) - Y_{i,t-1}(D_{i,t-1}^{\infty}) \mid C_i = c, A_i = 0] \\
=&E[Y_{i,t}(D_{i,t}^{\infty}) - Y_{i,t-1}(D_{i,t-1}^{\infty}) \mid C_i = \max\{C_i\}, A_i = 1] \\
-&E[Y_{i,t}(D_{i,t}^{\infty}) - Y_{i,t-1}(D_{i,t-1}^{\infty}) \mid C_i = \max\{C_i\}, A_i = 0].
\end{align*}
\end{Assumption}
Then, we have the following theorem.
\begin{Theorem}\label{sec3.theorem3}
If Assumptions \ref{sec3as1}-\ref{sec3as5} and \ref{sec3as9}-\ref{sec3as11} hold, each triple Wald-DID estimand $w^{DID}_{c,l,m}$ captures the $CLATT_{c,c+l}$; that is,
\begin{align*}
w^{DID}_{c,l,m}=E[Y_{i,c+l}(1)-Y_{i,c+l}(0)|C_i=c, A_i=1,CM_{c,c+l}],
\end{align*}
for each $c \in \{2,\dots,\max\{C_i\}-1\}$ and $l \in \{0,\max\{C_i\}-1-c\}$.
\end{Theorem}
\begin{proof}
This theorem follows from the similar argument in the proof of Theorem \ref{sec3.theorem2}. Therefore, we omit it for brevity.
\end{proof}
\begin{Remark}
When the treatment is non-binary and ordered, we have the following theorem.
\begin{Theorem}\label{sec3.theorem5}
\begin{itemize}
\item [(i)] If Assumptions \ref{sec3as1}-\ref{sec3as8} is satisfied, which replace the binary treatment with the non-binary one, each triple Wald-DID estimand $w^{DID}_{c,l}$ captures the $CACRT_{c,c+l}$; that is,
\begin{align*}
w^{DID}_{c,l}=CACRT_{c,c+l},
\end{align*}
for each $c\in \{2,\dots,T\}$ and $l \in \{0,\dots,T-c\}$.
\item [(ii)] If Assumptions \ref{sec3as1}-\ref{sec3as5} and \ref{sec3as9}-\ref{sec3as11} is satisfied, which replace the binary treatment with the non-binary one, each triple Wald-DID estimand $w^{DID}_{c,l,m}$ captures the $CACRT_{c,c+l}$; that is,
\begin{align*}
w^{DID}_{c,l,m}=CACRT_{c,c+l},
\end{align*}
for each $c \in \{2,\dots,\max\{C_i\}-1\}$ and $l \in \{0,\max\{C_i\}-1-c\}$.
\end{itemize}
\begin{proof}
This theorem holds by combining the proof in Theorems \ref{sec3.theorem2}-\ref{sec3.theorem3} with the proof in Theorem 4 in \cite{Miyaji2023}. Thus, we omit it for brevity.
\end{proof}
\end{Theorem}
\end{Remark}

\section{Estimation and inference}\label{sec4}
In this section, we describe the estimation and inference in triple DID-IV design.\par
In two periods settings considered in Section \ref{sec2}, we can estimate the triple Wald-DID estimand $w_{DID}$ by its sample analog, which we denote $\hat{w}_{DID}$:
\begin{align*}
\hat{w}_{DID}=\frac{\widehat{DID}_{Y,C=1,A=1}-\widehat{DID}_{Y,C=0,A=1}
-(\widehat{DID}_{Y,C=1,A=0}-\widehat{DID}_{Y,C=0,A=0})}{\widehat{DID}_{D,C=1,A=1}-\widehat{DID}_{D,C=0,A=1}
-(\widehat{DID}_{D,C=1,A=0}-\widehat{DID}_{D,C=0,A=0})},
\end{align*}
where
\begin{align*}
&\widehat{DID}_{Y,C=c,A=a}=\frac{E_N[(Y_{i,1}-Y_{i,0})\cdot \mathbf{1}\{C_i=c,A_i=a\}]}{E_N[\mathbf{1}\{C_i=c,A_i=a\}]},\\
&\widehat{DID}_{D,C=c,A=a}=\frac{E_N[(D_{i,1}-D_{i,0})\cdot \mathbf{1}\{C_i=c,A_i=a\}]}{E_N[\mathbf{1}\{C_i=c,A_i=a\}]},
\end{align*}
for $a \in \{0,1\}$ and $c \in \{0,1\}$. Here, $E_N[\cdot]$ is the sample analog of the expectation $E[\cdot]$, and $\mathbf{1}\{\cdot\}$ is the indicator function.\par
The following theorem presents the asymptotic property of $\hat{w}_{DID}$. 
\begin{Theorem}
\label{sec4.theorem4}
Suppose Assumptions \ref{sec2as1}-\ref{sec2as7} hold. Then, the triple Wald-DID estimator $\hat{w}_{DID}$ is consistent and asymptotically normal for the LATET.
\begin{align*}
\sqrt{n}(\hat{w}_{DID}-LATET) \xrightarrow{d} \mathcal{N}(0,V(\psi_{i})),
\end{align*}
where $\psi_{i}$ is influence function for $\hat{w}_{DID}$ and defined in Equation $\eqref{Appendix_theorem4_inf}$ in Appendix.
\begin{proof}
See Appendix.
\end{proof}
Note that if we have a non-binary, ordered treatment, we can replace $LATET$ with $ACRT$.\par
\end{Theorem} 
In practice, one can estimate $\hat{w}_{DID}$ and its standard error by the following IV regression:
\begin{align*}
Y_{i,t}&=\beta_{0}+\beta_{1}\mathbf{1}_{C_i=1}+\beta_{2}\mathbf{1}_{T_i=1}+\beta_{3}\mathbf{1}_{A_i=1}+\beta_{4}\mathbf{1}_{C_i=1,T_i=1}+\beta_{5}\mathbf{1}_{C_i=1,A_i=1}+\beta_{6}\mathbf{1}_{T_i=1,A_i=1}\\
&+\beta_{IV}D_{i,t}+\epsilon_{i,t}.
\end{align*}
The first stage regression is:
\begin{align*}
D_{i,t}&=\pi_{0}+\pi_{1}\mathbf{1}_{C_i=1}+\pi_{2}\mathbf{1}_{T_i=1}+\pi_{3}\mathbf{1}_{A_i=1}+\pi_{4}\mathbf{1}_{C_i=1,T_i=1}+\pi_{5}\mathbf{1}_{C_i=1,A_i=1}+\pi_{6}\mathbf{1}_{T_i=1,A_i=1}\\
&+\pi_{7}\mathbf{1}_{C_i=1,T_i=1,A_i=1}+\eta_{i,t},
\end{align*}
where $\mathbf{1}_{A}$ is the  indicator function and takes one if $A$ is true. $T_i \in \{0,1\}$ is time indicator and takes one if unit $i$ is in time $t=1$. Note that this IV regression runs the triple DID regression in both the first stage and the reduced form. Therefore, the IV estimator $\hat{\beta}_{IV}$ is equal to the triple Wald-DID estimator $\hat{w}_{DID}$.\par
In multiple period settings considered in Section \ref{sec3}, we can estimate the triple Wald-DID estimand $w^{DID}_{c,l}$ and $w^{DID}_{c,l,m}$ by its sample analog, which we denote $\hat{w}^{DID}_{c,l}$ and $\hat{w}^{DID}_{c,l,m}$, respectively:
\begin{align*}
&\hat{w}^{DID}_{c,l}=\frac{\widehat{DID}^{l}_{Y,C=c,A=1}-\widehat{DID}^{l}_{Y,C=\infty,A=1}-(\widehat{DID}^{l}_{Y,C=c,A=0}-\widehat{DID}^{l}_{Y,C=\infty,A=0})}{\widehat{DID}^{l}_{D,C=c,A=1}-\widehat{DID}^{l}_{D,C=\infty,A=1}-(\widehat{DID}^{l}_{D,C=c,A=0}-\widehat{DID}^{l}_{D,C=\infty,A=0})},\\
&\hat{w}^{DID}_{c,l,m}=\frac{\widehat{DID}^{l}_{Y,C=c,A=1}-\widehat{DID}^{l}_{Y,m,A=1}-(\widehat{DID}^{l}_{Y,C=c,A=0}-\widehat{DID}^{l}_{Y,m,A=0})}{\widehat{DID}^{l}_{D,C=c,A=1}-\widehat{DID}^{l}_{D,m,A=1}-(\widehat{DID}^{l}_{D,C=c,A=0}-\widehat{DID}^{l}_{D,m,A=0})},
\end{align*}
where
\begin{align*}
&\widehat{DID}^{l}_{Y,C=c,A=a}=\frac{E_N[(Y_{i,c+l}-Y_{i,c-1})\cdot \mathbf{1}\{C_i=c,A_i=a\}]}{E_N[\mathbf{1}\{C_i=c,A_i=a\}]},\\
&\widehat{DID}^{l}_{D,C=c,A=a}=\frac{E_N[(D_{i,c+l}-D_{i,c-1})\cdot \mathbf{1}\{C_i=c,A_i=a\}]}{E_N[\mathbf{1}\{C_i=c,A_i=a\}]},\\
&\widehat{DID}^{l}_{Y,m,A=a}=\frac{E_N[(Y_{i,c+l}-Y_{i,c-1})\cdot \mathbf{1}\{C_i=\max\{C_i\},A_i=a\}]}{E_N[\mathbf{1}\{C_i=\max\{C_i\},A_i=a\}]},\\
&\widehat{DID}^{l}_{D,m,A=a}=\frac{E_N[(D_{i,c+l}-D_{i,c-1})\cdot \mathbf{1}\{C_i=\max\{C_i\},A_i=a\}]}{E_N[\mathbf{1}\{C_i=\max\{C_i\},A_i=a\}]}.
\end{align*}
The following theorem presents the asymptotic property of $\hat{w}^{DID}_{c,l}$ and $\hat{w}^{DID}_{c,l,m}$.
\begin{Theorem}
\label{sec4.theorem5}
\begin{itemize}
\item[(i)] Suppose Assumptions \ref{sec3as1}-\ref{sec3as8} hold. Then, each triple Wald-DID estimator $\hat{w}^{DID}_{c,l}$ is consistent and asymptotically normal for the $CLATT_{c,c+l}$.
\begin{align*}
\sqrt{n}(\hat{w}^{DID}_{c,l}-CLATT_{c,c+l}) \xrightarrow{d} \mathcal{N}(0,V(\psi_{i,c,l})),
\end{align*}
where $\psi_{i,c,l}$ is influence function for $\hat{w}^{DID}_{c,l}$ and defined in Equation $\eqref{Appendix_theorem5_inf1}$ in Appendix.
\item[(ii)] Suppose Assumptions \ref{sec3as1}-\ref{sec3as5} and \ref{sec3as9}-\ref{sec3as11} hold. Then, each triple Wald-DID estimator $\hat{w}^{DID}_{c,l,m}$ is consistent and asymptotically normal for the $CLATT_{c,c+l}$.
\begin{align*}
\sqrt{n}(\hat{w}^{DID}_{c,l,m}-CLATT_{c,c+l}) \xrightarrow{d} \mathcal{N}(0,V(\psi_{i,c,l,m})),
\end{align*}
where $\psi_{i,c,l,m}$ is influence function for $\hat{w}^{DID}_{c,l,m}$ and defined in Equation $\eqref{Appendix_theorem5_inf2}$ in Appendix.
\end{itemize}
\begin{proof}
See Appendix.
\end{proof}
\end{Theorem}
Note that if we have a non-binary, ordered treatment, we can replace $CLATT_{c,c+l}$ with $CACRT_{c,c+l}$.\par
In staggered instrument settings, one can estimate $\hat{w}^{DID}_{c,l}$ ($\hat{w}^{DID}_{c,l,m}$) in two steps. First, we subset the data that contain only two cohorts and two periods, i.e., cohort $c$ and $\infty$ ($\max\{C_i\}$) and period $c+l$ and $c-1$. Next, in each data set, we run the following IV regression.
\begin{align*}
Y_{i,t}&=\beta^{c,l}_{0}+\beta^{c,l}_{1}\mathbf{1}_{C_i=c}+\beta^{c,l}_{2}\mathbf{1}_{T^{c,l}_i=c+l}+\beta^{c,l}_{3}\mathbf{1}_{A_i=1}+\beta^{c,l}_{4}\mathbf{1}_{C_i=c,T^{c,l}_i=c+l}+\beta^{c,l}_{5}\mathbf{1}_{C_i=c,A_i=1}\\
&+\beta^{c,l}_{6}\mathbf{1}_{T_i^{c,l}=c+l,A_i=1}+\beta^{c,l}_{IV}D_{i,t}+\epsilon^{c,l}_{i,t}.
\end{align*}
The first stage regression is:
\begin{align*}
D_{i,t}&=\pi^{c,l}_{0}+\pi^{c,l}_{1}\mathbf{1}_{C_i=c}+\pi^{c,l}_{2}\mathbf{1}_{T_i^{c,l}=c+l}+\pi^{c,l}_{3}\mathbf{1}_{A_i=1}+\pi^{c,l}_{4}\mathbf{1}_{C_i=c,T_i^{c,l}=c+l}+\pi^{c,l}_{5}\mathbf{1}_{C_i=c,A_i=1}\\
&+\pi^{c,l}_{6}\mathbf{1}_{T_i^{c,l}=c+l,A_i=1}+\pi^{c,l}_{7}\mathbf{1}_{C_i=c,T_i^{c,l}=c+l,A_i=1}+\eta^{c,l}_{i,t},
\end{align*}
where $T_i^{c,l} \in \{c-1,c+l\}$ is the time variable and takes $c+l$ if unit $i$ is in time $t=c+l$. Then, the IV estimator $\beta^{c,l}_{IV}$ corresponds to $\hat{w}^{DID}_{c,l}$ ($\hat{w}^{DID}_{c,l,m}$). We can also calculate the standard error by using the influence function derived in Theorem \ref{sec4.theorem5}.
\section{Conclusion}\label{sec5}
In this paper, we formalize a triple instrumented difference-in-differences. In this design, our target parameter is the local average treatment effect on the treated (LATET) and the identifying assumptions mainly comprise a monotonicity assumption and common acceleration assumptions in the treatment and the outcome. We show that in this design, the triple Wald-DID estimand, which scales the DDD estimand of the outcome by the DDD estimand of the treatment, captures the LATET. We extend the canonical triple DID-IV design to staggered instrument settings, and describe the estimation and inference in practice.
\appendix
\section*{Appendix}
In the proof of Theorems \ref{sec2.theorem1}-\ref{sec3.theorem2}, we omit the index $i$ to ease the notation. 
\begin{proofoftheorem}{1}
\end{proofoftheorem}
\begin{proof}
Note that Assumption \ref{sec2as5} guarantees that the triple Wald-DID estimand is well defined.\par
First, we consider the numerator of the triple Wald-DID estimand. Given Assumptions \ref{sec2as0}-\ref{sec2as7}, we have
\begin{align}
&E[Y_{1}-Y_{0}|C=1,A=1]-E[Y_{1}-Y_{0}|C=0,A=1]\notag\\
-&(E[Y_{1}-Y_{0}|C=1,A=0]-E[Y_{1}-Y_{0}|C=0,A=0])\notag\\
=&E[Y_{1}(D_1((0,1)))-Y_{0}(D_0((0,1)))|C=1,A=1]\notag\\
-&E[Y_{1}(D_1((0,0)))-Y_{0}(D_0((0,0)))|C=0,A=1]\notag\\
-&E[Y_{1}(D_1((0,0)))-Y_{0}(D_0((0,0)))|C=1,A=0]\notag\\
+&E[Y_{1}(D_1((0,0)))-Y_{0}(D_0((0,0)))|C=0,A=0]\notag\\
=&E[Y_{1}(D_1((0,1)))-Y_{0}(D_0((0,0)))|C=1,A=1]\notag\\
-&E[Y_{1}(D_1((0,0)))-Y_{0}(D_0((0,0)))|C=0,A=1]\notag\\
-&E[Y_{1}(D_1((0,0)))-Y_{0}(D_0((0,0)))|C=1,A=0]\notag\\
+&E[Y_{1}(D_1((0,0)))-Y_{0}(D_0((0,0)))|C=0,A=0]\notag\\
=&E[Y_1(D_1((0,1)))-Y_1(D_1((0,0)))|C=1,A=1]\notag\\
+&E[Y_{1}(D_1((0,0)))-Y_{0}(D_0((0,0)))|C=1,A=1]\notag\\
-&E[Y_{1}(D_1((0,0)))-Y_{0}(D_0((0,0)))|C=0,A=1]\notag\\
-&E[Y_{1}(D_1((0,0)))-Y_{0}(D_0((0,0)))|C=1,A=0]\notag\\
+&E[Y_{1}(D_1((0,0)))-Y_{0}(D_0((0,0)))|C=0,A=0]\notag\\
=&E[(D_1((0,1))-D_1((0,0)))\cdot(Y_1(1)-Y_1(0))|C=1,A=1]\notag\\
\label{Appendix_eq1}
=& LATET \cdot Pr(CM^{Z}|C=1,A=1).
\end{align}
Here, the first equality follows from Assumptions \ref{sec2as1}-\ref{sec2as2}. The second equality follows from Assumption \ref{sec2as4}. The third equality follows from Assumption \ref{sec2as7}. The final equality follows from Assumption \ref{sec2as3}.\par
Next, we consider the denominator of the triple Wald-DID estimand. Given Assumptions \ref{sec2as0}-\ref{sec2as7}, we have
\begin{align}
&E[D_{1}-D_{0}|C=1,A=1]-E[D_{1}-D_{0}|C=0,A=1]\notag\\
-&(E[D_{1}-D_{0}|C=1,A=0]-E[D_{1}-D_{0}|C=0,A=0])\notag\\
=&E[D_{1}((0,1))-D_{0}((0,1))|C=1,A=1]\notag\\
-&E[D_{1}((0,0))-D_{0}((0,0))|C=0,A=1]\notag\\
-&E[D_{1}((0,0))-D_{0}((0,0))|C=1,A=0]\notag\\
+&E[D_{1}((0,0))-D_{0}((0,0))|C=0,A=0]\notag\\
=&E[D_{1}((0,1))-D_{0}((0,0))|C=1,A=1]\notag\\
-&E[D_{1}((0,0))-D_{0}((0,0))|C=0,A=1]\notag\\
-&E[D_{1}((0,0))-D_{0}((0,0))|C=1,A=0]\notag\\
+&E[D_{1}((0,0))-D_{0}((0,0))|C=0,A=0]\notag\\
=&E[D_{1}((0,1))-D_{1}((0,0))|C=1,A=1]\notag\\
+&E[D_{1}((0,0))-D_{0}((0,0))|C=1,A=1]\notag\\
-&E[D_{1}((0,0))-D_{0}((0,0))|C=0,A=1]\notag\\
-&E[D_{1}((0,0))-D_{0}((0,0))|C=1,A=0]\notag\\
+&E[D_{1}((0,0))-D_{0}((0,0))|C=0,A=0]\notag\\
=&E[D_{1}((0,1))-D_{1}((0,0))|C=1,A=1]\notag\\
\label{Appendix_eq2}
=&Pr(CM^{Z}|C=1,A=1).
\end{align}
The second equality follows from Assumption \ref{sec2as4}. The third equality follows from Assumption \ref{sec2as6}. The final equality follows from Assumption \ref{sec2as3}.\par 
Combining \eqref{Appendix_eq1} with \eqref{Appendix_eq2}, we obtain the desirable result.
\end{proof}
\begin{proofoftheorem}{2}
\end{proofoftheorem}
\begin{proof}
Fix $c \in \{2,\dots,T\}$ and $l \in \{0,\dots,T-c\}$. Note that Assumption \ref{sec3as5} ensures that the triple Wald-DID estimand $w^{DID}_{c,l}$ is well defined.\par
First, we consider the numerator of $w^{DID}_{c,l}$. Given Assumptions \ref{sec3as1}-\ref{sec3as7}, we have
\begin{align}
&E[Y_{c+l}-Y_{c-1}|C=c,A=1]-E[Y_{c+l}-Y_{c-1}|C=\infty,A=1]\notag\\
-&(E[Y_{c+l}-Y_{c-1}|C=c,A=0]-E[Y_{c+l}-Y_{c-1}|C=\infty,A=0])\notag\\
=&E[Y_{c+l}(D_{c+l}^{c})-Y_{c-1}(D_{c-1}^{c})|C=c,A=1]\notag\\
-&E[Y_{c+l}(D_{c+l}^{\infty})-Y_{c-1}(D_{c-1}^{\infty})|C=\infty,A=1]\notag\\
-&E[Y_{c+l}(D_{c+l}^{\infty})-Y_{c-1}(D_{c-1}^{\infty})|C=c,A=0]\notag\\
+&E[Y_{c+l}(D_{c+l}^{\infty})-Y_{c-1}(D_{c-1}^{\infty})|C=\infty,A=0]\notag\\
=&E[Y_{c+l}(D_{c+l}^{c})-Y_{c-1}(D_{c-1}^{\infty})|C=c,A=1]\notag\\
-&E[Y_{c+l}(D_{c+l}^{\infty})-Y_{c-1}(D_{c-1}^{\infty})|C=\infty,A=1]\notag\\
-&E[Y_{c+l}(D_{c+l}^{\infty})-Y_{c-1}(D_{c-1}^{\infty})|C=c,A=0]\notag\\
+&E[Y_{c+l}(D_{c+l}^{\infty})-Y_{c-1}(D_{c-1}^{\infty})|C=\infty,A=0]\notag\\
=&E[Y_{c+l}(D_{c+l}^{c})-Y_{c+l}(D_{c+l}^{\infty})|C=c,A=1]\notag\\
+&E[Y_{c+l}(D_{c+l}^{\infty})-Y_{c-1}(D_{c-1}^{\infty})|C=c,A=1]\notag\\
-&E[Y_{c+l}(D_{c+l}^{\infty})-Y_{c-1}(D_{c-1}^{\infty})|C=\infty,A=1]\notag\\
-&E[Y_{c+l}(D_{c+l}^{\infty})-Y_{c-1}(D_{c-1}^{\infty})|C=c,A=0]\notag\\
+&E[Y_{c+l}(D_{c+l}^{\infty})-Y_{c-1}(D_{c-1}^{\infty})|C=\infty,A=0]\notag\\
=&E[(D_{c+l}^{c}-D_{c+l}^{\infty})\cdot (Y_{c+l}(1)-Y_{c+l}(0))|C=c,A=1]\notag\\
\label{Appendix_eq3}
=& CLATT_{c,c+l} \cdot Pr(CM_{c,c+l}|C=c,A=1).
\end{align}\par
The first equality follows from Assumptions \ref{sec3as1}-\ref{sec3as3}. The second equality follows from Assumption \ref{sec3as5}. The third equality follows from Assumption \ref{sec3as8}. The final equality follows from Assumption \ref{sec3as4}.\par
Next, we consider the denominator of $w^{DID}_{c,l}$. Given Assumptions \ref{sec3as1}-\ref{sec3as7}, we have
\begin{align}
&E[D_{c+l}-D_{c-1}|C=c,A=1]-E[D_{c+l}-D_{c-1}|C=\infty,A=1]\notag\\
-&(E[D_{c+l}-D_{c-1}|C=c,A=0]-E[D_{c+l}-D_{c-1}|C=\infty,A=0])\notag\\
=&E[D_{c+l}^{c}-D_{c-1}^{c}|C=c,A=1]\notag\\
-&E[D_{c+l}^{\infty}-D_{c-1}^{\infty}|C=\infty,A=1]\notag\\
-&E[D_{c+l}^{\infty}-D_{c-1}^{\infty}|C=c,A=0]\notag\\
+&E[D_{c+l}^{\infty}-D_{c-1}^{\infty}|C=\infty,A=0]\notag\\
=&E[D_{c+l}^{c}-D_{c-1}^{\infty}|C=c,A=1]\notag\\
-&E[D_{c+l}^{\infty}-D_{c-1}^{\infty}|C=\infty,A=1]\notag\\
-&E[D_{c+l}^{\infty}-D_{c-1}^{\infty}|C=c,A=0]\notag\\
+&E[D_{c+l}^{\infty}-D_{c-1}^{\infty}|C=\infty,A=0]\notag\\
=&E[D_{c+l}^{c}-D_{c+l}^{\infty}|C=c,A=1]\notag\\
+&E[D_{c+l}^{\infty}-D_{c-1}^{\infty}|C=c,A=1]\notag\\
-&E[D_{c+l}^{\infty}-D_{c-1}^{\infty}|C=\infty,A=1]\notag\\
-&E[D_{c+l}^{\infty}-D_{c-1}^{\infty}|C=c,A=0]\notag\\
+&E[D_{c+l}^{\infty}-D_{c-1}^{\infty}|C=\infty,A=0]\notag\\
\label{Appendix_eq4}
=& Pr(CM_{c,c+l}|C=c,A=1).
\end{align}\par
The first equality follows from Assumptions \ref{sec3as1}. The second equality follows from Assumption \ref{sec3as5}. The third equality follows from Assumption \ref{sec3as7}. The final equality follows from Assumption \ref{sec3as4}.\par
Combining \eqref{Appendix_eq3} with \eqref{Appendix_eq4}, we obtain the desirable result.
\end{proof}

\begin{proofoftheorem}{4}
\end{proofoftheorem}
\begin{proof}
First, we prove that $\hat{w}_{DID}$ is consistent for the LATET. By the Law of Large Numbers and continuous mapping theorem, we have $\hat{w}_{DID} \xrightarrow[]{p} w_{DID}$. Then, from Theorem \ref{sec2.theorem1}, we have $w_{DID}=LATET$ under Assumptions \ref{sec2as1}-\ref{sec2as7}.\par
Next, we prove that $\hat{w}_{DID}$ is asymptotically normal, deriving its influence function. We use the following fact, which is also found in \cite{De_Chaisemartin2018-xe}. 
\begin{Fact}
If
\begin{align*}
\sqrt{n}(\hat{A}-A)=\frac{1}{\sqrt{n}}\sum_{i=1}^n a_i+o_p(1),\hspace{3mm}\sqrt{n}(\hat{B}-B)=\frac{1}{\sqrt{n}}\sum_{i=1}^n b_i+o_p(1),
\end{align*}
we obtain
\begin{align*}
\sqrt{n}\left(\frac{\hat{A}}{\hat{B}}-\frac{A}{B}\right)=\frac{1}{\sqrt{n}}\sum_{i=1}^n\frac{a_i-(A/B)b_i}{B}+o_p(1).
\end{align*}
Using the above fact repeatedly, we have
\begin{align*}
\sqrt{n}(\hat{w}_{DID}-LATET)=\frac{1}{\sqrt{n}}\sum_{i=1}^{n}\psi_i+o_p(1),
\end{align*}
where $\psi_i$ is the influence function and takes the following form:
\begin{align}
\psi_i&=\frac{1}{DID_{D,C=1,A=1}-DID_{D,C=0,A=1}
-(DID_{D,C=1,A=0}-DID_{D,C=0,A=0})}\notag\\
\label{Appendix_theorem4_inf}
&\times \{\zeta_{i,1,1}-\zeta_{i,0,1}-\zeta_{i,1,0}+\zeta_{i,0,0}\}.
\end{align}
Here, we define $\delta_{i}=(Y_{i,1}-Y_{i,0})-w_{DID}\cdot(D_{i,1}-D_{i,0})$ and
\begin{align*}
\zeta_{i,c,a}=\frac{\mathbf{1}\{C_i=c,A_i=a\}\cdot[\delta_{i}-E[\delta_{i}|C_i=c,A_i=a]]}{\mathbf{1}\{C_i=c,A_i=a\}}.
\end{align*}
\end{Fact}

\end{proof}
\begin{proofoftheorem}{5}
\end{proofoftheorem}
\begin{proof}
\begin{itemize}
\item [(i)] 
Fix $c \in \{2,\dots,T\}$ and $l \in \{0,\dots,T-c\}$.\par
First, we show that $\hat{w}^{DID}_{c,l}$ is consistent for the $CLATT_{c,c+l}$. By the Law of Large Numbers and continuous mapping theorem, we have $\hat{w}^{DID}_{c,l} \xrightarrow[]{p} w^{DID}_{c,l}$. Then, from Theorem \ref{sec3.theorem2}, we have $w^{DID}_{c,l}=CLATT_{c,c+l}$ under Assumptions \ref{sec3as1}-\ref{sec3as8}.\par
Next, we show that $\hat{w}^{DID}_{c,l}$ is asymptotically normal, deriving its influence function. By the similar argument in the proof of Theorem \ref{sec4.theorem4}, we have:
\begin{align*}
\sqrt{n}(\hat{w}^{DID}_{c,l}-CLATT_{c,c+l})=\frac{1}{\sqrt{n}}\sum_{i=1}^{n}\psi_{i,c,l}+o_p(1),
\end{align*}
where $\psi_{i,c,l}$ is the influence function and takes the following form:
\begin{align}
\psi_{i,c,l}&=\frac{1}{DID^{l}_{D,C=c,A=1}-DID^{l}_{D,C=\infty,A=1}-(DID^{l}_{D,C=c,A=0}-DID^{l}_{D,C=\infty,A=0})}\notag\\
\label{Appendix_theorem5_inf1}
&\times \{\zeta^{l}_{i,c,1}-\zeta^{l}_{i,c,0}-\zeta^{l}_{i,\infty,1}+\zeta^{l}_{i,\infty,0}\}.
\end{align}
Here, we define $\delta_{i,c,l}=(Y_{i,c+l}-Y_{i,c-1})-w^{DID}_{c,l}\cdot(D_{i,c+l}-D_{i,c-1})$ and
\begin{align*}
\zeta^{l}_{i,c,a}=\frac{\mathbf{1}\{C_i=c,A_i=a\}\cdot[\delta_{i,c,l}-E[\delta_{i,c,l}|C_i=c,A_i=a]]}{\mathbf{1}\{C_i=c,A_i=a\}},
\end{align*}
for $c \in \{2,\dots,T,\infty\}$ and $a \in \{0,1\}$.
\item [(ii)] This theorem follows from the similar argument in (i). The influence function for the $\hat{w}^{DID}_{c,l,m}$, which we denote $\psi_{i,c,l,m}$, takes the following form:
\begin{align}
&\psi_{i,c,l,m}\notag
\\
=&\frac{1}{DID^{l}_{D,C=c,A=1}-DID^{l}_{D,C=\max\{C_i\},A=1}-(DID^{l}_{D,C=c,A=0}-DID^{l}_{D,C=\max\{C_i\},A=0})}\notag\\
\label{Appendix_theorem5_inf2}
&\times \{\zeta^{l}_{i,c,1}-\zeta^{l}_{i,c,0}-\zeta^{l}_{i,\max\{C_i\},1}+\zeta^{l}_{i,\max\{C_i\},0}\}.
\end{align}
\end{itemize}
\end{proof}
\bibliographystyle{econ-econometrica.bst}
\bibliography{reference} 
\end{document}